\documentclass[11pt, reqno]{amsart}

\IfFileExists{mymtpro2.sty}{%
  \usepackage[subscriptcorrection]{mymtpro2}
}{}

\usepackage{a4, upgreek, amssymb, amsmath}

\usepackage{appendix}

\newtheorem{theorem}{Theorem}[section]
\newtheorem{lemma}{Lemma}[section]
\newtheorem{corollary}{Corollary}[section]
\newtheorem{proposition}{Proposition}[section]
\theoremstyle{definition}
\newtheorem{remark}[theorem]{Remark}
\newtheorem{remarks}[theorem]{Remarks}

\newcommand{\labelnummer}{\mbox{\normalfont (\roman{numcount})}}%

\makeatletter

  {\let\curlabelspeicher\@currentlabel%
    \begin{list}{\labelnummer}%
      {\usecounter{numcount}\leftmargin0pt%
        \topsep0.5ex\partopsep2ex\parsep0pt\itemsep0ex\@plus1\p@%
        \labelwidth2.5em\itemindent3.5em\labelsep1em%
      }%
    \let\saveitem\item%
    \def\item{\saveitem%
      \def\@currentlabel{{\upshape\curlabelspeicher}$\,$\labelnummer}}%
    \let\savelabel\label%
    \def\label##1{\savelabel{##1}%
      \@bsphack%
        \ifmmode\else%
          \protected@write\@auxout{}%
          {\string\newlabel{##1item}{{\labelnummer}{\thepage}}}%
        \fi%
      \@esphack%
    }%
  }{\end{list}}%

\renewcommand{\appendix}{\def\thesection{\textsc{Appendix}}}

 \let\leq\le
 \let\geq\ge

 \let\Im\undefined

\DeclareMathOperator{\Im}{Im}

\DeclareMathOperator{\tr}{tr\kern1pt}

\makeatletter

\newif\ifper\pertrue
\def\per{.}

\def\bti{\@ifnextchar[\bbti\bbbti}
\def\bbti[#1]#2{#2, #1.}
\def\bbbti#1{#1.}

\def\z{\@ifnextchar[\zz\zzz}
\def\zz[#1]#2#3#4#5{\perfalse\emph{#2} \textbf{#3}, #4 (#5) [#1]}
\def\zzz#1#2#3#4{\emph{#1} \textbf{#2}, #3 (#4)\ifper\per\fi\pertrue}

\def\pub{\@ifstar\pubstar\pubnostar}
\def\pubnostar{\@ifnextchar[\@@pubnostar\@pubnostar}
\def\@@pubnostar[#1]#2#3#4{#2, #3, #4, #1\ifper\per\fi\pertrue}
\def\@pubnostar#1#2#3{#1, #2, #3\ifper\per\fi\pertrue}
\def\pubstar[#1]#2#3#4{\perfalse #2, #3, #4 [#1]\pertrue}

\makeatother

\newcommand{\beq}{\begin{equation}}
\newcommand{\eeq}{\end{equation}}
\newcommand{\ba}{\begin{array}}
\newcommand{\ea}{\end{array}}
\newcommand{\bea}{\begin{eqnarray}}
\newcommand{\eea}{\end{eqnarray}}

\newcommand{\R}{\mathbb{R}}

\newcommand{\Z}{\mathbb{Z}}

\newcommand{\N}{\mathbb{N}}

\newcommand{\C}{\mathbb{C}}

\def\P{I\kern-.30em{P}}
\def\E{I\kern-.30em{E}}
\renewcommand{\E}{\mathbb{E}\mkern2mu}
\renewcommand{\P}{\mathbb{P}}

\newcommand{\Schr}{Schr\"odinger}

{

\begin{document}

\title[Local eigenvalue statistics for random band matrices]{On the local eigenvalue statistics for random band matrices in the localization regime}

\author[P.\ D.\ Hislop]{Peter D.\ Hislop}
\address{Department of Mathematics,
    University of Kentucky,
    Lexington, Kentucky  40506-0027, USA}
\email{peter.hislop@uky.edu}

\author[M.\ Krishna]{M.\ Krishna}
\address{Ashoka University, 
Plot No 2, Rajiv Gandhi Education City,
Rai, Haryana 131029, India}
\email{krishna.maddaly@ashoka.edu.in}


\begin{abstract}
We study the local eigenvalue statistics $\xi_{\omega,E}^N$ associated with the eigenvalues of one-dimensional, $(2N+1) \times (2N+1)$ random band matrices with independent, identically distributed, real random variables and band width growing as $N^\alpha$, for $0 < \alpha < \frac{1}{2}$. We consider the limit points associated with the random variables $\xi_{\omega,E}^N [I]$, for $I \subset \R$, and $E \in (-2,2)$. For Gaussian distributed random variables with $0 \leq \alpha < \frac{1}{7}$, we prove that this family of random variables has nontrivial limit points 
for almost every $E \in (-2,2)$, and that these limit points are Poisson distributed with positive intensities. The proof is based on an analysis of the characteristic functions of the random variables $\xi_{\omega,E}^N [I]$ and associated quantities related to the intensities, as $N$ tends towards infinity, and employs known localization bounds of \cite{schenker, psss}, and the strong Wegner and Minami estimates \cite{psss}. Our more general result applies to random band matrices with random variables having absolutely continuous distributions with bounded densities. Under the hypothesis that the localization bounds hold for 
$0 < \alpha < \frac{1}{2}$, we prove that any nontrivial limit points of the random variables $\xi_{\omega,E}^N [I]$ are distributed according to Poisson distributions.
\end{abstract}

\maketitle \thispagestyle{empty}

\tableofcontents



\section{Random band matrices: Statement of the problem and main results}\label{sec:rbm1}
\setcounter{equation}{0}

A  random band matrix (RBM) in one dimension, $H_L^N$, of size $2N+1$ and band width $W := 2L + 1$, for an integer $L = \lfloor N^\alpha \rfloor$,  
with $0 \leq \alpha \leq 1$, is a $(2N+1)\times (2N+1)$ real, symmetric matrix
defined through its matrix elements as
\bea\label{eq:RBMdefn1}
\langle e_i,  H^N_L e_j \rangle & = & \frac{1}{\sqrt{2L+1}}\left\{\begin{array}{ccc}
\omega_{ij} & {\rm if} & |i-j| \leq L \\
0  & {\rm if} & |i-j| > L\end{array}\right. ,
\eea
with
$$
-N \leq i,j \leq N .
$$
The results here also hold for periodic band matrices for which the norm in \eqref{eq:RBMdefn1} is replaced by periodic norm $|i-j |_1$. 

The real random variables $\omega_{ij}= \omega_{ji}$ within the band are independent and identically distributed (iid) up to symmetry. The random variables are assumed to have mean zero, variance one, and finite moments. 
These include the most common case of a Gaussian distribution for which we assume $\omega_{ij} \in \mathcal{N} (0,1)$. 

The normalization in \eqref{eq:RBMdefn1} is chosen so that the variances $\sigma_{ij} := 
\E \{ |\langle e_i,  H^N_L e_j \rangle |^2 \}$, satisfy 
$$ 
\sum_{j = -L}^{L} \sigma_{ij} = 1 = \sum_{i = -L}^{L} \sigma_{ij}.
$$
That is, the sum of the variances in each row and in each column is equal to one. 

We denote by $\{ E_j^N(\omega) \}_{j=\N}^N$ the set of the $2N+1$ eigenvalues of $H_\omega^N$.
The local  eigenvalue statistics (LES) is defined with respect to the rescaled eigenvalues of $H_\omega^N$ defined by $\widetilde{E}_j(\omega) := N( E_j^N(\omega) - E_0)$ for any $E_0 \in (-2,2)$. The LES centered at $E_0$ is the weak limit as $N \rightarrow \infty$ of the process
\beq\label{eq:les1}
\xi_\omega^N (ds) := \sum_{j=-N}^N \delta (   N( E_j^N(\omega) - E_0) - s) ~ds.
\eeq

\subsection{Density of states}\label{subsec:dos1}

For $0 < \alpha < 1$, the integrated density of states (IDS) for RBM is given by the semi-circle law:
\beq\label{eq:scl1}
{N}_{sc} (E) = \frac{1}{2 \pi} \int_{- 2}^E ~ \sqrt{ ( 4 - s^2)_+} ~ds, ~~ E \in [-2,2] .
\eeq
For $\alpha = 0$, the IDS is not semi-circle but has a remainder behaving like $\mathcal{O}(W^{-1})$. These results were proved by Bogachev, Molchanov, and Pastur \cite{bmp}, using the method of moments, and in   Molchanov, Pastur, and Khorunzhi\u{i} \cite{mpk}, using Green's functions. The proof for the case of case of $\alpha = 1$ is due to Wigner, and we refer the reader to \cite{mehta}.
The density of states function (DOSf) is given by
\beq\label{eq:scl2}
n_{sc}(E) =  \frac{1}{2 \pi} ~ \sqrt{ ( 4 - E^2)_+}, ~~ E \in [-2,2] .
\eeq
For a measurable subset $J \subset \R$, the semi-circle measure of $J$ is denoted by 
\beq\label{eq:scl3}
N_{sc}(J) = \int_J n_{sc}(s) ~ds .
\eeq


\subsection{Conjectures for the local eigenvalue statistics of RBM}\label{subsec:conjectures1}

There are two main conjectures about the behavior of the LES for RBM as the exponent $\alpha$ varies $0 \leq \alpha \leq 1$:
\begin{itemize}
\item Localization regime: $0 \leq \alpha < \frac{1}{2}$ and the LES at $E \in (-2,2)$ are given by a Poisson point process with intensity measure $n_{\rm SC}(E) ~ds$, where $n_{\rm SC}(s)$ is the density of the semi-circle law. 
\item Delocalization regime: $\frac{1}{2} < \alpha \leq 1$ and the LES is that of the Gaussian orthogonal ensemble (GOE). 
\end{itemize}
These conjectures originated with the numerical studies in \cite{cmi}. Analytical evidence for these conjectures was presented in \cite{fm} based on the analysis of a related $\sigma$-model.  

One way to understand these conjectures is to note that according to the localization bound in \eqref{eq:locHypRBM1}, the localization length for scale $N$ behaves like $\ell^N_{loc} = \mathcal{O}(N^{\alpha \mu})$. Consequently, the ratio of the localization length to the overall scale is 
\beq\label{eq:locLenght1}
\kappa_N := \frac{N^{\alpha \mu}}{N} = N^{\alpha \mu - 1} .
\eeq
For the assumed optimal value $\mu = 2$, we see that this ratio $\kappa_N \rightarrow 0$, if $\alpha < \frac{1}{2}$, and $\kappa \rightarrow \infty$, if $\alpha > \frac{1}{2}$. This is reminiscent of the critical behavior observed for the scaled disorder model of 1D random {\Schr}  operators \cite{kvv}.

In this note, we prove that, under two hypotheses, a weaker version of the first conjecture is true. These hypotheses are satisfied for Gaussian random variables $\omega_{ij}$. For other distributions, our proof establishes the first conjecture only for $0 < \alpha < \frac{1}{3}$ under the strong Wegner estimate and the weak Minami estimate. 

We mention related work of Shcherbina and Shcherbina \cite{ss2014} who proved that the LES for the complex Gaussian Hermitian case and $\alpha < \frac{1}{2}$ could not be GUE by analyzing the second mixed moment of the characteristic polynomial using the supersymmetric method.    

The main problem of the LES for RBM in the localization phase is the determination of the intensity of the limiting process. For an interval $I \subset \R$, and an energy $E \in (-2,2)$, we define 
\beq\label{eq:b_N_defn1}
b_N(I,E) := \E \left\{  {\rm Tr} P_{H_L^N} \left( \frac{1}{N} I + E \right)  \right\}.
\eeq
The intensity measure is given by the limit:
\beq\label{eq:localExp001}
\lim_{N \rightarrow \infty} b_N(I,E) = \lim_{N \rightarrow \infty} \E \left\{  {\rm Tr} P_{H_L^N} \left( \frac{1}{N} I + E \right)  \right\}.
\eeq
Although we strongly expect this limit to be $n_{sc}(E) |I|$, so that the intensity measure of the limiting process is $n_{sc}(E) ~ds$, we have not succeeded in proving this under hypotheses $[H1]$ and $[H2]$. Instead, we prove that for any bounded interval $I \subset \R$, there is a set of energies $E \in (-2,2)$ of full measure for which the random variables $\xi_{\omega,E}^N[I]$ have limit points that are Poisson distributed with a nontrivial intensity.  These nontrivial intensities are the finite, positive limit points of $\mathcal{B}_{I,E} := \{ b_N(I,E) ~|~ N \in \N \}$. We define this set as $\mathcal{L}_{I,E} := \{0< p(I,E) < \infty ~|~ p((I,E) ~{\rm is ~a ~limit ~point ~of} ~\mathcal{B}_{I,E} \}$. 


\subsection{The main results}\label{subsec:main1}

We first state our main result on RBM with real Gaussian random entries.

\begin{theorem}\label{thm:GaussianCase1}
Let $H_L^N$ be a random band matrix as defined in \eqref{eq:RBMdefn1}, with entries that are real, independent, Gaussian random variables $\omega_{ij} = \omega_{ji} \in {\mathcal{N}}(0,1)$, and with band width $2L+1$, where $L = \lfloor N^\alpha \rfloor$, for $0 \leq \alpha < \frac{1}{7}$.  Then,  
 for any interval $I \subset \R$, there exists a set of energies $ \Omega_I \subset (-2,2)$ of full measure, so that for any $E \subset \Omega_I$, the set of random variables $\{ \xi_{\omega,E}^N [I] , N \in \N \}$, associated with the local eigenvalue statistics, has limit points that are Poisson distributed random variables.  In particular, the set of non-trivial intensities $\mathcal{L}_{I,E} \neq \emptyset$, for almost every $E \in (-2,2)$. 
\end{theorem}

%

The main new contribution is that the set $\{ b_N(I,E) ~|~ E \in \Omega_I, n \in \N \}$, where $b_N(I,E)$ is defined in \eqref{eq:b_Nconv1}, has at least one finite, positive limit point. Theorem \ref{thm:GaussianCase1} is a specific application of our main theorems on local eigenvalue statistics for RBM in the localization regime. In order to discuss the general case, we present two hypotheses and then discuss models for which these hypotheses hold true. 

\begin{enumerate}

%
%
%

\item \textbf{strong Wegner and Minami estimates} 


\vspace{.1in}

\noindent
\textbf{[H1s]:} 
The following estimates hold at any scale $\widetilde{N}$. 


\begin{enumerate}
\item $[sW]:$ For any bounded interval $I \subset \R$, we have  
\beq\label{eq:wegner_gauss1}
\E \{ {\rm Tr} \chi_I (H_L^{\widetilde{N}} ) \} \leq C_0 |I| \widetilde{N} .
\eeq

\medskip

\item $[sM]$: For any bounded interval $I \subset \R$, we have  
\bea\label{eq:minami_gauss1} 
\P \{ {\rm Tr} \chi_I(H_L^{\widetilde{N}} )  \geq 2 \} & \leq & \E \{ {\rm Tr} \chi_I(H_L^{\widetilde{N}} ) ( {\rm Tr} \chi_I(H_L^{\widetilde{N }}) - 1 ) \} \nonumber \\
 & \leq  & C_M (  |I| \widetilde{N} )^2 .
\eea
\end{enumerate}

\item \textbf{weak Wegner and Minami estimates} 

\vspace{.1in}

\noindent
\textbf{[H1w]:} 
The following estimates hold at any scale $\widetilde{N}$. 
\begin{enumerate}
\item $[wW]:$ For any bounded interval $I \subset \R$, we have  
\beq\label{eq:wegner_weak1}
\E \{ {\rm Tr} \chi_I(H_L^{\widetilde{N}} ) \} \leq C_0 |I| \sqrt{W} \widetilde{N} .
\eeq

\medskip

\item $[wM]$: For any bounded interval $I \subset \R$, we have  
\bea\label{eq:minami_weak1} 
\P \{ {\rm Tr} \chi_I(H_L^{\widetilde{N}} )  \geq 2 \} & \leq & \E \{ {\rm Tr} \chi_I(H_L^{\widetilde{N}} ) ( {\rm Tr} \chi_I(H_L^{\widetilde{N}} ) - 1 ) \} \nonumber \\
 & \leq  & C_M (  |I| \sqrt{W} {\widetilde{N}} )^2 .
\eea
\end{enumerate}

\item \textbf{Localization estimate}

\vspace{.1in}

\noindent
\textbf{[H2]}:  For $0 \leq \mu \leq 2$, the following estimate holds.  
Given $\rho > 0$ and $s \in (0,1)$, there exist finite constants $C_{\rho,s} > 0$ and $\alpha_{\rho,s} > 0$,
so that for all $j, k \in \Lambda_N$, there exists a $\sigma > 0$ so that 
\beq\label{eq:locHypRBM1}
\E \left\{| \langle \delta_j,  (H_\omega ^N - E )^{-1} \delta_k  \rangle |^s   \right\} 
\leq C_{\rho,s} N^{{s \alpha \sigma}} e^{- \alpha_{\rho,s} \frac{|j-k|}{N^{\alpha \mu} } }   , 
\eeq
for all $E \in [-\rho, \rho]$. 


\end{enumerate}

\medskip

\begin{remarks}
\begin{enumerate}


\item The distinction between the weak and strong estimates in $[H1s]$ and $[H1w]$ is the factor of $\sqrt{W} \sim N^{\frac{\alpha}{2}}$. The weak estimates are obtained by spectral averaging applied to the diagonal random variables (see, for example, \cite{bh1}). The strong estimates for Gaussian random variables are due to \cite{psss}.

\item With regard to the localization bound $[H2]$, since we want exponential decay outside of the band width for $|j-k| \approx N$, we must have $\alpha \mu < 1$. 
If we assume that $[H2]$ holds for $\mu = 2$, then we must have $\alpha < \frac{1}{2}$, the conjectured  regime of localization. 

\item The localization bound \eqref{eq:locHypRBM1} was proven to hold in \cite{schenker} for a family of random variables with an absolutely continuous density and satisfying other conditions. Unfortunately, the proof in these cases only guarantees the existence of $\mu > 0$ and $\sigma > 0$.
For the case of $\mathcal{N}(0,1)$-Gaussian random variables, Schenker proved that the estimate holds for  some $\mu \leq 8$ and $\sigma = \frac{1}{2}$.  This means that the exponent $0 \leq \alpha < \frac{1}{8}$. This result was improved in  \cite[Theorem 4]{psss} to some $\mu \leq 7$ so that $0 \leq \alpha < \frac{1}{7}$. The localization bound is believed to hold up to the critical exponent $\alpha = \frac{1}{2}$.

\end{enumerate}

\end{remarks}

Not much is known about the nature of LES even for the range $0< \alpha < \frac{1}{7}$, for which the localization bound has been proven (\cite{psss, schenker}). The analysis of the characteristic exponents in section \ref{sec:limit_pts1}, and of the intensity in section \ref{sec:intensity1}, together with the results of \cite{psss, schenker}, form the basis of Theorem \ref{thm:GaussianCase1} that may be paraphrased as: 
\noindent
\emph{Consider real random band matrices with Gaussian distributed random entries as in \eqref{eq:RBMdefn1}, and with band widths growing as $N^\alpha$, for $0 < \alpha < \frac{1}{7}$. For any interval $I \subset \R$, 
there is a set of energies $\Omega_I \subset  (-2,2)$ of full measure such that :
\begin{enumerate}
\item All nontrivial limit points of the random variables $\{ \xi_{\omega,E}^N[I] ~|~ N \in \N \}$ are Poisson distributed;
\item For each energy $E \in \Omega_I$, there are nontrivial, Poisson distributed limit points of
$\{ \xi_{\omega,E}^N[I] ~|~ N \in \N \}$.
\end{enumerate}
}
\noindent
Nontriviality means that the limit point is random variable with a finite, nonzero characteristic exponent.

In the two main theorems below, we show how characterizations of the LES may be derived from various assumptions. For example, we believe that the localization bound should hold in the natural range 
$0 \leq \alpha < \frac{1}{2}$. Assuming this, we prove that the results paraphrased above hold for $\alpha$ in this natural range.  


We begin with a theorem that is rather general and which applies under the weakest possible hypotheses: The weak Wegner estimate $[H1w]$, the weak Minami estimate $[H1w]$, and the localization bound $[H2]$. This result states that the nontrivial limit points of the random variables $\{ \xi_{\omega,E}^N [I] ~|~ N \in \N \}$ are distributed according to Poisson distributions. Theorem \ref{thm:main0} is similar to our result \cite[Theorem 5.1]{hk1} on the LES for random {\Schr} operators on $L^2 (\R^d)$. 


\begin{theorem}\label{thm:main0}
Let $H_L^N$ be a random band matrix with band width $2L+1$ as defined in \eqref{eq:RBMdefn1} with $L = \lfloor N^\alpha \rfloor$, for $0 \leq \alpha \leq 1$.   Then, the weak Wegner estimate $[H1w]$ and the weak Minami estimate $[wM1]$ both hold. We assume the localization estimate  $[H2]$ for $\mu =2$ and $0 < \alpha <  \frac{1}{2}$. 
 Then, for each $E \in (-2,2)$, all the nontrivial limit points of the random variables $\{ \xi_{\omega,E}^N [I] ~|~ N \in \N \}$  are distributed according to Poisson distributions with characteristic exponents having the form 
 $$
 \Psi (t) = (e^{it} - 1) p_1^*(I),
 $$
 where the measures $p_1^*(I)$, defined in terms of the local array $\{ \eta_{\omega,E}^{p,N_k} \}$ (see section \ref{subsec:outline1}), are the nontrivial limit points of the family
 $$
 \{  N^{1-\beta} \P \{ \eta_{\omega,E}^{1,N} [I] = 1 \} ; N \in \N \} .
 $$
 \end{theorem}
 
%

The weak Wegner and Minami estimates for RBM follow easily from spectral averaging over the diagonal entries and standard methods. The only limitation on the width comes from the localization bounds. In the following corollary, we strengthen Theorem \ref{thm:main0} using what is presently known concerning the localization bounds from \cite{psss,schenker}.

\begin{corollary} If the real random variables in the RBM \eqref{eq:RBMdefn1} are Gaussian distributed and $0 < \alpha < \frac{1}{7}$, then the assumptions of Theorem \ref{thm:main0} hold. Consequently, the nontrivial limit points of $\{ \xi_{\omega,E}^N [I] ~|~ N \in \N \}$  are distributed according to Poisson distributions. More generally, if the random variables are distributed with a bounded density, then there exists a nonzero $0 < \alpha_0 < \frac{1}{2}$, so that the assumptions and results of Theorem \ref{thm:main0} hold for $0 < \alpha < \alpha_0$.
\end{corollary}

Theorem \ref{thm:main0} and its corollary do not state the existence of any nontrivial limit points. Upon strengthening the hypotheses, our second main result is the following theorem. 

\begin{theorem}\label{thm:main1}
Let $H_L^N$ be a random band matrix with band width $2L+1$ as defined in \eqref{eq:RBMdefn1} with $L = \lfloor N^\alpha \rfloor$, for $0 \leq \alpha \leq 1$. Let $I \subset \R$ be a bounded interval.   
We assume the strong Wegner and strong Minami estimates of $[H1s]$, and the localization estimate  $[H2]$ for $\mu =2$ and $0 < \alpha < \frac{1}{2}$. 
Then, there exists a set $\Omega_I \subset (-2,2)$, depending on $I$, with $|\Omega_I| = 4$, such that for fixed $E \in \Omega_I$, the random variables $\{ \xi_{\omega,E}^N [I] ~|~ N \in \N \}$ have non-trivial Poisson-distributed limit points and the intensity of the corresponding Poisson distribution is given
by $\limsup_N b_N(I,E) > 0$, where $b_N(I,E)$ is defined in \eqref{eq:localExp01}.
In particular, each finite, positive limit point of the set $\{ b_N(I,E) ~|~ E \in \Omega_I, n \in \N \}$
is the intensity of a Poisson distributed random variable that is a limit point of the set $\{ \xi_{\omega,E}^N [I] ~|~ N \in \N \}$. 
\end{theorem}

We recall that a stronger result for the $\alpha = 0$ fixed band width case was obtained in \cite{bh1}. The local point process $\xi_{\omega,E}^N$ converges to a Poisson point process with intensity measure $n_{\infty,W}(E) ~ds$,
where $n_{\infty,W}$ is the density of states given by $n_{\infty,W}(I) = n_{sc}(I)  + \mathcal{O}(W^{-1})$ for any interval $I \subset \R$. Since the band width is independent of $N$, the strong and weak Wegner and Minami estimates are the same and a basic localization bound holds at all energies. The stronger result for $\alpha = 0$ is due to the fact that much more can be proved about the convergence of the density of states $n_N$ in this case.  

%

Another immediate corollary follows if we replace the strong Minami estimate $[H1s]$ by the weak Minami estimate $[H1w]$. The constraint on $\alpha$ is due to the condition $\alpha + \beta < 1$ in  \eqref{eq:charactExp6} and condition (3) in Remark \ref{remark:scales1}.

\begin{corollary}\label{cor:weakM1}
We assume the strong Wegner estimate of $[H1s]$, the weak Minami estimate of $[H1w]$, and the localization bound $[H2]$ with $\mu = 2$.  Then the results of Theorem \ref{thm:main1} hold
for $0 < \alpha < \frac{1}{3}$. 
\end{corollary}

Finally, we give a sufficient condition for the Poisson distribution of the limit points of $\xi_{\omega,E}^N[I]$ if we only use hypotheses $[H1w]$ and $[H2]$. We do not know how to prove the necessary estimates in order to establish a finite, nonvanishing intensity, under these weaker conditions.


\begin{proposition}\label{prop:mainW1}
Let $H_L^N$ be a random band matrix with band width $2L+1$ as defined in \eqref{eq:RBMdefn1} with $L = \lfloor N^\alpha \rfloor$, for $0 \leq \alpha \leq 1$.   
We assume the hypotheses $[H1w]$, the weak Wegner and the weak Minami estimates, and the localization estimate 
$[H2]$ for $\mu =2$ and $0 < \alpha < \frac{1}{3}$.
 Let $I \subset \R$ be a bounded interval and $E \in (-2,2)$. 
 Then, the random variables $\{ \xi_{\omega,E}^N [I] ~|~ N \in \N \}$ have non-trivial Poisson-distributed limit points provided $\limsup_N N_\beta \P \{ \eta_{\omega,E}^{1,N}[I]  = 1  \} > 0$ and finite. The intensity of the corresponding Poisson distribution is given
by $\limsup_N b_N(I,E) > 0$, where $b_N (I,E)$ is defined in \eqref{eq:localExp01}.
\end{proposition}


\subsection{Brief outline of the proof}\label{subsec:outline1}

The localization hypothesis $[H2]$ is used to recast the problem in terms of an array of independent random variables. 
As usual, we divide the set $\{ -N, -N+1, \ldots, -1, 0, 1, \ldots, N-1, N \}$ into subsets of length $2 \lfloor N^\beta \rfloor + 1$, for $0< \alpha < \beta < 1$. We always assume $2N+1$ is  divisible by $2\lfloor N^\beta \rfloor$. We label each subset of size $2 \lfloor N^\beta \rfloor + 1$-points by $p = 1, 2, \ldots, N_\beta$, where $N_\beta := (2N+1)(2\lfloor N^\beta \rfloor+1)^{-1}$. 

We associate a RBM $H_L^{p,N}$, of width $W = 2 \lfloor N^\alpha \rfloor +1$, for each such $p$. Using the eigenvalues of $H_L^{p,N}$ we construct the local eigenvalue statistics $\eta_{\omega,E}^{N,p}$ as in \eqref{eq:les1} using the scaling by $N$. The process $\zeta_{\omega,E}^N$ is a superposition of independent processes $\eta_{\omega,E}^{N,p}$.
We assume that $\alpha < \frac{1}{2}$. If the weak Minami estimate is used, we further assume that 
$\alpha + \beta < 1$. 
 
The proof consists of the following steps. These steps are an adaptation of the arguments of \cite{hk1} to the RBM models. We fix a bounded interval $I \subset \R$ and $E \in (-2,2)$. 

\begin{enumerate}

\item The localization bound $[H2]$ implies that the family of random variable $\zeta_{\omega,E}^N [I] = \sum_{p} \eta_{\omega,E}^{N,p} [I]$ has the same limit points as $\xi_{\omega,E}^N [I]$. As a consequence, the limit points of $\zeta_{\omega,E} [I]$ and $\xi_{\omega,E} [I]$ are infinitely-divisible random variables. These are described by their characteristic functions.     

\item The Minami estimate, either weak  $[H1w]$ or strong $[H1s]$, guarantees that the distributions of the limit points of the local random variables $\zeta_{\omega,E}^N [I]$ have no double points. This determines the form of the characteristic exponents. If the associated intensity is non-zero, then the limit points are Poisson distributed. 

\item The strong Wegner estimate $[H1s]$ guarantees that some of the limit points of 
$\zeta_{\omega,E}^N [I]$, and consequently of $\xi_{\omega,E}^N [I]$, are Poisson distributed with positive intensity. 


\end{enumerate}


\subsection{Contents}\label{subsec:contents}

In section \ref{sec:limit_pts1}, we describe the characteristic functions associated with the random variables 
$\zeta_{\omega,E}^N [I]$. We use the Wegner and Minami estimates in order to describe the form of the characteristic exponent. The corresponding intensity of the distribution is studied in detail in section \ref{sec:intensity1}. We prove that the intensity is positive for the distribution of at least some of the limit points, establishing Theorem \ref{thm:main1}. Section \ref{sec:limitPP1} presents the main steps of the proof of the equality of the limit points of  $\zeta_{\omega,E} [I]$ and $\xi_{\omega,E} [I]$. 
 

\section{Properties of the characteristic functions of $\zeta_{\omega,E}^{N}[I]$}\label{sec:limit_pts1}

\setcounter{equation}{0}

We follow the approach of \cite[section 5]{hk1} in order to obtain an expression for the characteristic function of the limit point random variables corresponding to $\zeta_{\omega,E}^N[I]$. We recall from section \ref{subsec:outline1} that the local process $\zeta_{\omega,E}^N$ is a superposition of independent processes $\eta_{\omega,E}^{N,p}$, for $p=1, \ldots, N_\beta$. The characteristic exponent $\Psi_N(t)$ of the random variable $\zeta_{\omega,E}^{N}[I]$ is defined by
\beq\label{eq:charactExp1}
\E \{ e^{it \zeta_{\omega,E}^N[I]} \} =: e^{\Psi_N (t)}.
\eeq
The characteristic function has the form
\bea\label{eq:charact1}
\E \{ e^{it \zeta_{\omega,E}^N[I]} \} & = & \Pi_{p=1}^{N_\beta} ~ \E \{ e^{it \eta_{\omega,E}^{p,N}[I]} \} \nonumber \\
  &=& e^{ \sum_{p=1}^{N_\beta} \log \E \{ e^{it \eta_{\omega,E}^{p,N}[I]} \} } ,
  \eea
  where $N_\beta := (2N+1) (2 \lfloor N^\beta \rfloor + 1)^{-1} \in \N$, for $0 < \alpha < \frac{1}{2}$ and $0 < \alpha < \beta < 1$. 
We expand the logarithm as 
\bea\label{eq:charact2}
 \left| \log \left[ \E \{ e^{it \eta_{\omega,E}^{p,N}[I]} -1 \} + 1 \right] \right| & = & \left|  \E \{ e^{it \eta_{\omega,E}^{p,N}[I]} -1 \}  \right|  \nonumber \\
  & & + \mathcal{O} \left( \left|  \E \{ e^{it \eta_{\omega,E}^{p,N}[I]} \} -1 \right|^2 \right) .
\eea
There are two possible estimates:
\begin{itemize}
\item The weak Wegner estimate $[H1w]$ implies that 
\bea\label{eq:charact3w}
\left|  \E \{ e^{it \eta_{\omega,E}^{p,N}[I]} -1 \}  \right| & \leq  & t \E \{ \eta_{\omega,E}^{p,N}[I] \} \nonumber \\ 
 & \leq & t N^{\frac{\alpha}{2} + \beta - 1} .
 \eea
 which vanishes as $N \rightarrow \infty$ under the condition $\alpha + \beta < 1$. This also justifies the expansion \eqref{eq:charact2} as
 \beq\label{eq:charact4w}
 \sum_{p=1}^{N_\beta} \left|  \E \{ e^{it \eta_{\omega,E}^{p,N}[I]} \} -1 \right|^2  \leq \frac{N}{N^\beta} \cdot \frac{N^{\alpha + 2\beta}}{N^2} = \frac{N^{\alpha + \beta}}{N},
 \eeq
 which also vanishes. 
  \item The strong Wegner estimate $[H1s]$ implies that 
\bea\label{eq:charact3s}
\left|  \E \{ e^{it \eta_{\omega,E}^{p,N}[I]} -1 \}  \right| & \leq  & t \E \{ \eta_{\omega,E}^{p,N}[I] \} \nonumber \\ 
 & \leq & t N^{\beta - 1} .
 \eea
 which vanishes as $N \rightarrow \infty$ under the condition $0 < \beta < 1$. This also justifies the expansion \eqref{eq:charact2} as
 \beq\label{eq:charact4s}
 \sum_{p=1}^{N_\beta} \left|  \E \{ e^{it \eta_{\omega,E}^{p,N}[I]} \} -1 \right|^2  \leq \frac{N}{N^\beta} \cdot \frac{N^{2\beta}}{N^2} = \frac{N^{\beta}}{N},
 \eeq
 which also vanishes. 
  \end{itemize}

Consequently, in either case, we can write the characteristic function as
\beq\label{eq:charact5}
\E \{ e^{it \zeta_{\omega,E}^N[I]} \} = e^{ \sum_{p=1}^{N_\beta} \E \{  e^{it \eta_{\omega,E}^{p,N}[I]} -1  \}  } ,
\eeq
up to vanishing terms.  Because of this, and the homogeneity in $p$, we may assume that the characteristic exponent $\Psi_N(t)$ of $\zeta_{\omega,E}^N[I]$ has the form
\beq\label{eq:charactExp2}
\Psi_N(t)  = N_\beta \E \{ ( e^{it \eta_{\omega,E}^{1,N}[I]} -1 ) \}   . 
\eeq

 

To complete the analysis of the limiting characteristic exponent, we write 
\bea\label{eq:charactExp4}
\Psi_N(t) &=& N_\beta  \E \{ ( e^{it \eta_{\omega,E}^{1,N}[I]} -1 ) \}  \nonumber \\
 &=& N_\beta \sum_{j=1}^\infty ( e^{itj} -1) \P \{ \eta_{\omega,E}^{1,N}[I]  = j  \} .
 \eea
We note that the general result, Theorem \ref{thm:main0}, requiring only the weak 
Wegner estimate of $[H1w]$, follows from this expression. Furthermore, the conditions that guarantee the vanishing of the expression in \eqref{eq:charact4w}, that is, $0 \leq \alpha + \beta <1$ and $\alpha < \beta$, require that $\alpha < \frac{1}{2}$. This shows that $\alpha < \frac{1}{2}$ is a natural condition for the limit points to be described by a Poisson distribution. 
 
 Proceeding with the proof of Theorem \ref{thm:main1}, we use the Minami estimates of $[H1]$. Writing 
\beq\label{eq:charactExp5}
  \sum_{j=1}^\infty ( e^{itj} -1) \P \{ \eta_{\omega,E}^{1,N}[I]  = j  \} =
    ( e^{it} -1) \P \{ \eta_{\omega,E}^{1,N}[I]  = 1  \} + \sum_{j=2}^\infty ( e^{itj} -1) \P \{ \eta_{\omega,E}^{1,N}[I]  = j  \}  ,
    \eeq
    the contribution in \eqref{eq:charactExp4} coming from the second term in \eqref{eq:charactExp5} may be bounded 
    \begin{itemize}
 \item Using the weak Minami estimate $[H1w]$,     
   \bea\label{eq:charactExp6}
 N_\beta \left|     \E \{ ( e^{it \eta_{\omega,E}^{1,N}[I]} -1 ) \chi_{\eta_{\omega,E}^{1,N}[I] \geq 2} \} \right| & \leq & 2 N_\beta \P \{ \eta_{\omega,E}^{1,N}[I] \geq 2 \} \nonumber \\
  & \leq & 2 |I|^2 N^{\alpha + \beta - 1},
  \eea
which vanishes as $N \rightarrow \infty$ as $\alpha + \beta < 1$, or
\item Using the strong Minami estimate $[H1s]$,
 \bea\label{eq:charactExp7}
 N_\beta \left|     \E \{ ( e^{it \eta_{\omega,E}^{1,N}[I]} -1 ) \chi_{\eta_{\omega,E}^{1,N}[I] \geq 2} \} \right| & \leq & 2 N_\beta \P \{ \eta_{\omega,E}^{1,N}[I] \geq 2 \} \nonumber \\
  & \leq & 2 |I|^2 N^{\beta - 1},
  \eea
which vanishes as $N \rightarrow \infty$ as $\beta < 1$.
\end{itemize}
Consequently, in either case we have
\beq\label{eq:charactExp8}
\lim_{N \rightarrow \infty} \Psi_N(t) =     ( e^{it} -1) p_1(I),
\eeq
where 
\beq\label{eq:weight1}
p_1(I) := \lim_{N \rightarrow \infty} N_\beta \P \{ \eta_{\omega,E}^{1,N}[I]  = 1  \}    ,
\eeq
provided the limit exists. The existence of this limit will be studied in the next section.

\section{Intensity of the distribution of the limit points of $\zeta_{\omega,E}^{N}[I]$}\label{sec:intensity1}

\setcounter{equation}{0}

The main result of this section is the calculation of the intensity of the limiting Poisson distribution for the limit points of the random variables $\xi_{\omega,E}^N [I]$, for any interval $I \subset \R$. We begin with two lemmas. As discussed in section \ref{sec:rbm1}, the calculation of the limit:
\beq\label{eq:localExp01}
p_1(I) = 
\lim_{N \rightarrow \infty} \E \left\{  {\rm Tr} P_{H_L^N} \left( \frac{1}{N} I + E \right)  \right\} ,
\eeq
is essential for proving the convergence of the local point process  $\xi_{\omega,E}^N$ to a Poisson point process. Although we do not calculate this limit here, we prove the existence of positive limit points of the sequence defined on the right side of \eqref{eq:localExp01}.  

To relate this calculation to the result \eqref{eq:weight1} of section \ref{sec:limit_pts1}, we note that the analog of \eqref{eq:localExp01} for the array of random variables $\{ \eta_{\omega,E}^{p,N} [I] \}$ is given by 
\beq\label{eq:localExpArray0}
\lim_{N \rightarrow \infty} \sum_{p=1}^{N_\beta} \E \left\{  {\rm Tr} P_{H_L^{p,N}} \left( \frac{1}{N} I + E \right)  \right\} .
\eeq
The weak Minami estimate $[H1w]$ implies the following :
\beq\label{eq:localExpArray1}
  \sum_{p=1}^{N_\beta} \E \left\{  {\rm Tr} P_{H_L^{p,N}} \left( \frac{1}{N} I + E \right)  \right\} 
= N_\beta \P \{ \eta_{\omega,E}^{1,N} [I] =1 \} + \mathcal{O}(N^{\alpha+ \beta - 1}) ,
\eeq
whereas the strong Minami estimate $[H1s]$ yields $ \mathcal{O}(N^{\beta - 1}) $, so we see that the limit in \eqref{eq:weight1} is equivalent to the limit in \eqref{eq:localExpArray0}. The localization hypothesis $[H2]$ implies in turn that the limits in \eqref{eq:localExpArray0} and \eqref{eq:localExp01} are the same, if they exist, and that they have the same limit points.

\begin{lemma}\label{lemma:intensity1}
For any bounded interval $I \subset \R$, and $E \in (-2,2)$, we define 
\beq\label{eq:localExp1}
b_N(I,E) := \E \left\{  {\rm Tr} P_{H_L^N} \left( \frac{1}{N} I + E \right)  \right\}.
\eeq
Then, for any interval $J \subset (-2,2)$, we have
\beq\label{eq:intensity1}
\lim_{N \rightarrow \infty}  \int_J b_N(I,E) ~dE = |I| N_{sc}(J) > 0   ,
\eeq
where $N_{sc}$ is the semi-circle DOSm defined in \eqref{eq:scl3}.
\end{lemma}

\begin{proof}

\noindent
1. The local density of states measure $\mu_N$ ($\ell$DOSm) is defined by 
\beq\label{eq:lDOSf1}
\mu_N(I) := \frac{1}{2N+1}  \E \left\{  {\rm Tr} P_{H_L^N} (I )  \right\} ,
\eeq
for measurable subsets $I \subset \R$. The Wegner estimate $[H1]$ implies that $\mu_N$ is absolutely continuous with respect to Lebesgue measure and its density $n_N$, the  local density of states function ($\ell$DOSf), satisfies
\beq\label{eq:ldosf1}
 \mu_N (I) = \int_I ~n_N(s) ~ds .
 \eeq 
By a change of variables, we write $b_N$, defined in \eqref{eq:localExp1}, in terms of the $\ell$DOSf:
\beq\label{eq:localExp2}
b_N(I,E) = N \mu_N \left( \frac{1}{N} I + E \right) = \int_I ~n_N \left( \frac{x}{N} + E \right) ~dx .
\eeq

\noindent
2. We choose any interval $J \subset (-2,2)$ and integrate $b_N$ over $J$:
\beq\label{eq:ave1}
\int_J ~b_N(I,E) ~dE = \int_J ~dE ~\int_I ~dx  ~n_N \left( \frac{x}{N} + E \right) .
\eeq
Since $n_N$ is smooth, and the integrals are over bounded sets, the order of integration may be exchanged and we define
\beq\label{eq:ave2}
b_N (x,J) := \int_J  ~n_N \left( \frac{x}{N} + E \right) ~dE .
\eeq
We now study the limit of $b_N(x,J)$ as $N \rightarrow \infty$. It follows from the work of \cite{mpk} that for $0 < \alpha \leq 1$,
\beq\label{eq:weakDOS1}
\lim_{N \rightarrow \infty} \mu_N(J) = N_{sc}(J).
\eeq
(For the $\alpha = 0$ case, there is an $\mathcal{O}(W^{-1})$-correction to the semi-circle law.)

\noindent
3. Given any $\epsilon > 0$, for any $0 < M < \infty$, there exists $N_{\epsilon,M}$ so that 
for any $N > N_{\epsilon, M}$, we have $|x / N| < \epsilon$, for any $x \in [-M, M]$. For $J = [c,d]$, and for any $x \in [-M,M]$, a change of variables in \eqref{eq:ave2} results in the bounds 
\beq\label{eq:inequal1}
\int_{c+ \epsilon}^{d - \epsilon} ~ n_N(s) ~ds \leq b_N(x,J) = \int_{\frac{x}{N} + J} ~n_N(s) ~ds \leq
\int_{c - \epsilon}^{d + \epsilon} ~ n_N(s) ~ds .
\eeq
It follows from \eqref{eq:weakDOS1}
that
\beq\label{eq:inequal2}
\lim_{N \rightarrow \infty} \int_{c+ \epsilon}^{d - \epsilon} ~ n_N(s) ~ds = N_{sc}([c +\epsilon, d - \epsilon]), 
 \eeq
 and similarly for the upper bound in \eqref{eq:inequal1}.
 Consequently,  for any $x \in [-M, M]$, relations \eqref{eq:weakDOS1}-\eqref{eq:inequal2} imply that
\beq\label{eq:inequal3}
 N_{sc}([c +\epsilon, d - \epsilon]) \leq \liminf_{N \rightarrow \infty} b_N(x,J) \leq \limsup_{N \rightarrow \infty} b_N(x,J)  \leq N_{sc}([c -\epsilon, d + \epsilon]) .
\eeq
Hence, since \eqref{eq:inequal3} holds for any $\epsilon > 0$,  we have the pointwise limit
\beq\label{eq:ptw1}
\lim_{N \rightarrow \infty} b_N (x,J) = N_{sc}(J),  
\eeq
for any $x \in [-M, M]$. 

\noindent
4. We next prove that the set of function $\{ b_N(x,J) ~|~ ~ x \in [-M, M] \}$ is uniformly bounded in $N$. 
As follows from \eqref{eq:inequal1}, that for $N > N_{\epsilon,M}$,
\beq\label{eq:unifBdd1}
\int_{c+ \epsilon}^{d - \epsilon} ~ n_N(s) ~ds
\leq  \inf_{x \in [-M,M]} b_N(x,J) \leq 
\sup_{x \in [-M,M]} b_N(x,J) \leq \int_{c - \epsilon}^{d + \epsilon} ~ n_N(s) ~ds .
\eeq
As above, computing liminf and limsup over $N  > N_{\epsilon,M}$, we obtain
\bea\label{eq:unifBdd2}
N_{sc} ([ {c+ \epsilon}, {d - \epsilon}]) & \leq &  \liminf_N  \left\{ \sup_{x \in [-M,M]} b_N(x,J) \right\} \nonumber \\
 &\leq & 
\limsup_N \left\{ \sup_{x \in [-M,M]} b_N(x,J)  \right\} \nonumber \\
  & \leq  &  N_{sc}( [{c - \epsilon}, {d + \epsilon} ])  , 
\eea
and similarly for $\inf_{x \in [-M,M]} b_N(x,J)$.
Since \eqref{eq:unifBdd2} holds for all $\epsilon > 0$, we obtain the result
\beq\label{eq:unifBdd3}
\lim_{N \rightarrow \infty} \left\{ \sup_{x \in [-M,M]} b_N(x,J) \right\} = N_{sc}(J). 
\eeq
The uniform boundedness of of the set $\{ b_N(x,J) ~|~ ~ x \in [-M, M] \}$ follows from this.

\noindent
5. A consequence of the pointwise convergence \eqref{eq:ptw1} and the uniform boundedness of 
$\{ b_N(x,J) ~|~ ~ x \in K \}$, for any compact subset $K \subset \R$, is that for any bounded interval $I \subset \R$, the Lebesgue Dominated Convergence Theorem implies that 
\beq\label{eq:main_limit1}
\lim_{N \rightarrow \infty} \int_I dx \int_J  dE ~n_N \left( \frac{x}{N} + E \right) = |I| N_{sc}(J).
\eeq
From \eqref{eq:ave1}, this means that
\beq\label{eq:ave3}
\lim_{N \rightarrow \infty} \int_J ~b_N(I,E) ~dE = |I| N_{sc}(J) = \int_I dx \int_J  dE ~n_{sc}(E).
\eeq
\end{proof}


In order to prove the existence of subsequences $\{ N_k \}$ so that $b_{N_k}(I,E)$ has a positive limit, we need control over the local density of states function $n_N(E)$, In the proof of the following lemma, we use the strong Wegner estimate $[H1s]$. This is the only place where this strong estimate is used. 

\begin{lemma}\label{lemma:subseq1}
Assume the strong Wegner estimate of $[H1s]$. For almost every $E \in (-2,2)$, depending on $I$, there exists a sequence $N_k(E) \rightarrow \infty$
so that 
\beq\label{eq:b_Nconv1}
\lim_{k \rightarrow \infty} b_{N_k(E)} (I,E) =: h(I,E) > 0.
\eeq
\end{lemma}

\begin{proof}
By the strong Wegner estimate of $[H2]$, it follows that there exists a finite $C_0 > 0$ so that $\| n_N \|_\infty \leq C_0$, for all integers $N > 0$. As a consequence, for  all $E \in (-2, 2)$, we have 
\bea\label{eq:Nconv2}
 b_N(I,E) & = & \int_I ~n_N \left( \frac{x}{N} + E \right)  \nonumber \\
  & = & N \int_{ \frac{I}{N} + E } ~n_N(u) ~du \nonumber \\
  & \leq & C_0 |I|.
  \eea
We proved in Lemma \ref{lemma:intensity1}, \eqref{eq:intensity1}, that for any interval $J \subset (-2,2)$, we have 
\beq\label{eq:intensity2}
\lim_{N \rightarrow \infty} \int_J ~dE ~b_N(I,E) = N_{sc}(J) |I|  >  0 .
\eeq
We now suppose that $\limsup_N b_N(I,E) = 0$, for any interval $J \subset (-2,2)$. 
Applying the reverse Fatou inequality to \eqref{eq:intensity2}, we obtain
\beq\label{eq:Ncov4}
0 < N_{sc}(J) |I| = \limsup_{N}  \int_J ~dE ~b_N(I,E) \leq \int_J ~dE ~ \limsup_N b_N(I,E) =0,
\eeq
a contradiction. Hence, $\limsup_N b_N(I,E) > 0$ and finite, for almost every $E \in (-2,2)$, and there exists a subsequence so that \eqref{eq:b_Nconv1} holds. 
\end{proof}

Hence, each finite, positive limit point $h(I,E) = \limsup_\N b_N(I,E)$ is the intensity of the Poisson distribution of a limit point of the set of random variables $\{ \xi_{\omega,E}^{N_k(E)}[I] ~|~ N \in \N \}$ for almost every $E \in (-2, 2)$. 


\section{Localization: Equality of the limit points of $\zeta_{\omega,E}^N[I]$ and $\xi_{\omega,E}^N[I]$}\label{sec:limitPP1}

\setcounter{equation}{0}

We sketch the proof of the key result of localization
\beq\label{eq:loc1}
\lim_{N \rightarrow \infty} ~\E \{ \xi_{\omega,E}^N[f]  - \zeta_{\omega,E}^N[f] \}  = 0,
\eeq
for real test function $f$. Following Minami \cite[section 2]{minami}, it suffices to prove \eqref{eq:loc1} for functions $f(x) = \Im (x-z)^{-1}$, for $\Im z > 0$. This leads to the consideration of the imaginary parts of the Green's functions $R_N(z) := (H_L^N - z)^{-1}$ and $R_{p,N}(z) := (H_L^{p,N} - z)^{-1}$. 

As above, we construct an array of independent point processes as follows. We choose $0 < \alpha < \beta < 1$, with $0 < \alpha <   \frac{1}{2}$, and, if the weak Minami estimate is used, $\alpha + \beta < 1$. We partition the set of integers $\langle-N,N\rangle := [-N,N] \cap \Z$ into non-overlapping ordered subsets $I_{\beta,N}(p)$
containing $2 \lfloor N^\beta \rfloor + 1$ points: 
\bea\label{eq:setDecomp1}
\langle -N,N \rangle & = & \bigcup_{p=1}^{N_\beta}   I_{\beta, N}(p)  \nonumber \\
 &=&  \langle  -N , -N +  ( 2 \lfloor N^\beta \rfloor)  \rangle  \cup  \langle  -N +(2 \lfloor N^\beta \rfloor) +1, -N + 2 ( 2 \lfloor N^\beta \rfloor)  \rangle \nonumber \\
  & &  
\bigcup_{p=3}^{N_\beta}    \langle  -N + (p-1)(2 \lfloor N^\beta \rfloor) + 1, -N + p ( 2 \lfloor N^\beta \rfloor)  \rangle
\eea
and where $N_\beta := \frac{ 2N+1}{2 \lfloor N^\beta \rfloor +1}$, assumed to be an integer,  is the number of these disjoint subsets.
The local eigenvalue point process associated with the local RBM $H_L^{p,N}$ and the subset  $I_{\beta,N}(p)$ is denoted by $\eta_{\omega,E}^{p,N}$. 

We make the following definitions:
\begin{itemize}
\item The \emph{end points} of the ordered set $I_{\beta,N}(p)$ are $I^\pm_{\beta, N}(p)$, with $I^-_{\beta,N}(p) < I^+_{\beta,N}(I)$. 

\item The \emph{boundary} of $I_{\beta,N}(p)$ is defined by $\partial I_{\beta,N}(p) := \{ j \in I_{\beta,N}(p)  ~|~ {\rm dist}(j, I_{\beta,N}(p)^\pm) \leq N^\alpha \}$. 

\item The \emph{interior} of $I_{\beta,N}(p)$ is defined by 
${\rm Int} I_{\beta,N}(p) := \{ j \in I_{\beta,N}(p)  ~|~ {\rm dist}(j, \partial I_{\beta,N}(p) > N^{\mu \alpha} \log N^\delta) \}$, where $ \delta > 0$ will be chosen below.

\end{itemize}
\noindent
Note that $| {\rm Int} I_{\beta,N}(p)| = \mathcal{O}(N^\beta)$, and $| \partial I_{\beta,N}(p) | = \mathcal{O}(N^\alpha)$, and 
$$
{\rm dist}( {\rm Int}I_{\beta,N}(p),  \partial I_{\beta,N}(p) ) = \mathcal{O}( N^{\mu \alpha} \log N^\delta) $$. 
\noindent

\begin{remark}\label{remark:scales1}
The various scales determined by the exponents $\alpha, \beta, \mu$ satisfy the relations:
\begin{enumerate}
\item $0 \leq \alpha < \frac{1}{2}$, and $\alpha < \beta$ to insure that $N^\alpha < N^\beta$;
\item $0 < \alpha + \beta < 1$, using the weak Wegner estimate, see \eqref{eq:charact4w}, or $0 < \beta < 1$, using the strong Wegner estimate,  see \eqref{eq:charact4s};
\item $\alpha \mu < \beta <1$ to ensure exponential decay from \eqref{eq:locHypRBM1} and that $N^{\mu \alpha} \log N^\delta < N^\beta$.
\end{enumerate}
For the conjectured optimal value $\mu =2$, and working with only the weak assumptions $[H1w]$,
the conditions are $2 \alpha < \beta < 1$ and $\alpha + \beta < 1$ (see \eqref{eq:charactExp6}). These are satisfied if we require that $\alpha < \frac{1}{3}$ and $\frac{2}{3} < \beta < 1$. If we use the strong assumptions, we may take $\alpha < \frac{1}{2}$.
\end{remark}

With the choice of $f(x) = \Im (x-z)^{-1}$, for $\Im z > 0$, the difference in \eqref{eq:loc1} is bounded as
\beq\label{eq:green1}
\frac{1}{2N+1} \left| \Im {\rm Tr} R_N(z) - \sum_{p=1}^{N_\beta} \Im {\rm Tr} R_{N,p}(z) \right| \leq A_N(z) + B_N(z) ,
\eeq
where
\beq\label{eq:green2}
A_N(z) := \frac{1}{2N+1} \sum_{p=1}^{N_\beta} \left( \sum_{j \in I_{\beta,N}(p) \backslash {\rm Int} I_{\beta,N}(p)}
\left[ \Im G_N(j,j;z) +  \Im G_{N,p}(j,j;z) \right] \right) ,
\eeq
and
\beq\label{eq:green3}
B_N(z) := \frac{1}{2N+1} \sum_{p=1}^{N_\beta} \sum_{j \in {\rm Int} I_{p,N}(I) }
\left[ \sum_{\langle k, \ell \rangle \in \partial I_{\beta,N}(p)}  ~| G_{N,p}(j,k;z) | 
|\omega_{k \ell}| |G_{N,p}(\ell,j;z)|  \right].
\eeq

We estimate $A_N(z)$ using \emph{a priori} bounds on the matrix elements of the resolvents:
\bea\label{eq:est1}
\E \{ A_N(z) \}  & \leq  & \frac{1}{2N+1} \left( \frac{N}{N^\beta} \right) ( N^{\alpha \mu} \log N^\delta )
 \left[ \E \{ | \Im G_N(j,j;z) | \}  + \E \{ |\ \Im G_{N,p}(j,j;z)| \}  \right] \nonumber \\
  & = &  \mathcal{O} \left( \frac{N^{\alpha \mu} \log N^\delta}{N^\beta} \right) ,
 \eea 
which vanishes as $N \rightarrow \infty$.

Turning to the second term $B_N(z)$, we have 
\beq\label{eq:est2}
\E \{ B_N(z)^\frac{s}{2} \}   \leq   \frac{1}{2N+1} \sum_{p=1}^{N_\beta} \sum_{j \in {\rm Int} I_{p,N}(I) }
\left[ \sum_{\langle k, \ell \rangle \in \partial I_{\beta,N}(p)}  ~ \E \{ | G_{N,p}(j,k;z) |^\frac{s}{2} 
\left|\frac{\omega_{k \ell}}{N^\frac{\alpha}{2}} \right|^\frac{s}{2} |G_{N,p}(\ell,j;z)|^\frac{s}{2} \}  \right].
\eeq
We use the Cauchy-Schwarz inequality to bound the expectation:
\bea\label{eq:est3}
\E \{ | G_{N,p}(j,k;z) |^\frac{s}{2} 
|\omega_{k \ell}|^\frac{s}{2} |G_{N,p}(\ell,j;z)|^\frac{s}{2} \}  & \leq  & 
\E \{ | G_{N,p}(j,k;z) |^s \}^\frac{1}{2} \E \{ |\omega_{k \ell}|^s |G_{N,p}(\ell,j;z)|^s \}^\frac{1}{2} \}  \nonumber \\
  & \leq & \E \{ | G_{N,p}(j,k;z) |^s \}^\frac{1}{2} 
\E \{ |\omega_{k \ell}|^{2s} \}^\frac{1}{4} \E \{  |G_{N,p}(\ell,j;z)|^{2s} \}^\frac{1}{4} \}  \nonumber \\
  & \leq & N^{\frac{s \alpha}{2}} e^{- \kappa_{p,s}N^{- \alpha \mu} |j-k|} E \{ |\omega_{k \ell}|^{2s} \}^\frac{1}{4} \E \{  |G_{N,p}(\ell,j;z)|^{2s} \}^\frac{1}{4}.  \nonumber \\
   & & 
  \eea
We have assumed that the moments of the random variables are bounded, and it follows from \eqref{eq:locHypRBM1} that the resolvent satisfies the bound 
\beq\label{eq:est4}
\E \{  |G_{N,p}(\ell,j;z)|^{2s} \}^\frac{1}{4} \leq C_1 N^{s \sigma},
\eeq
for a constant $C_1 > 0$ independent of $N$ and $z \in \C^+$, and for $s < \frac{1}{2}$, so that we obtain
\bea\label{eq:est5}
\E \{ B_N(z)^\frac{s}{2} \} &  \leq &  C_1  N^{\frac{s \alpha}{4}+ s \sigma -\beta} \sum_{j \in {\rm Int} I_{\beta,N}(1) } \left( \sum_{\langle k, \ell \rangle \in \partial I_{\beta,N}(1)} e^{- \kappa_{1,s} \frac{|j-k|}{N^{ \alpha \mu} }}  \right)
\nonumber \\
 & \leq & C_1  N^{\frac{ \alpha}{8}+ \frac{\sigma}{2}-\beta} N^\beta N^{\alpha \mu} \left[ \frac{1}{N^{\kappa_{1,s} \delta}} - e^{- N^{\beta - \alpha \mu}} \right].
\eea
This vanishes as $N \rightarrow \infty$ provided we choose $\delta > \kappa_{1,s}^{-1} [\alpha  \left( \frac{1}{8} + \mu \right) + \frac{\sigma}{2}]$ and $\beta > \alpha \mu$. 
The remainder of the proof follows as in the proof of Minami in \cite[section 2]{minami}.



\begin{thebibliography} {[10]}
\frenchspacing \baselineskip=12 pt plus 1pt minus 1pt

\bibitem{bmp}  L.\ V.\  Bogachev, S.\ A.\  Molchanov, L.\ A.\  Pastur: \emph{On the level density of random band matrices}, Mat.\ Zametki \textbf{50} (1991), no.\ 6, 31–42, 157; translation in Math.\ Notes \textbf{50} (1991), no.\ 5-6, 1232–1242 (1992).



\bibitem{bourgade} P.\ Bourgade:  {\it Random band matrices}, Proceedings of the International Congress of Mathematicians—Rio de Janeiro 2018. Vol. IV. Invited lectures, 2759–2784, World Sci. Publ., Hackensack, NJ, 2018; arXiv:1807.03031.

\bibitem{byy1} P.\  Bourgade, H.-T.\  Yau,  J.\  Yin:\emph{ Random band matrices in the delocalized phase, I: quantum unique ergodicity and universality},  Comm.\ Pure Appl.\ Math. \textbf{73} (2020), no.\ 7, 1526–1596; arXiv:1807.01559. 

\bibitem{byyy1} P.\ Bourgade, F.\ Yang, H.-T.\ Yau,  J.\ Yin:  \emph{Random band matrices in the delocalized phase, II: generalized resolvent estimates.} J.\  Stat.\ Phys.\textbf{ 174 } (2019), no.\ 6, 1189–1221. 


\bibitem{bh1} B.\ Brodie, P.\ D.\ Hislop: \emph{The density of states and local eigenvalue statistics for random band matrices of fixed width}; arXiv:2008.13167.

\bibitem{cmi} G.\ Casati, L.\ Giulio, F.\ Izrailev: \emph{Scaling properties of band random matrices.} Phys.\ 
Rev.\ Lett. \textbf{64} (1990), no.\ 16, 1851–1854.

 \bibitem{fm} Y.\ V.\  Fyodorov, A.\ D.\ Mirlin: \emph{Scaling properties of localization in random band matrices: a $\sigma$-model approach},
Phys.\ Rev.\ Lett. \textbf{67} (1991), no.\ 18, 2405-2409.

\bibitem{hk1} P.\ D.\ Hislop, M.\ Krishna: \emph{Eigenvalue statistics random Schr\"odinger operators
with non rank one perturbations}, {Commun.\ Math.\ Phys.} {\bf 340} (2015), no.\ 1, 125--143; arXiv:1409.2328.

\bibitem{kvv} E.\ Kritchevski, B.\ Valk\'o, B.\ Vir\'ag: \emph{The scaling limit of the critical one-dimensional random \Schr operator,}  Commun.\ Math.\ Phys. \textbf{314} (2012), no.\ 3, 775-806; MR2964774.


\bibitem{mehta} M.\ L.\ Mehta: \emph{Random matrices}, third edition, Pure and applied mathematics series vol.\ \textbf{142}, Elsivier: Amsterdam, 2004. 

\bibitem{minami} N.\ Minami:  {\it Local fluctuation of the spectrum of a multidimensional Anderson tight binding model}, Commun.\ Math.\ Phys. {\bf 177}, 709-725 (1996).

\bibitem{mpk} S.\ A.\ Molchanov, L.\  A.\ Pastur,  A.\ M.\  Khorunzhi\u{i}: \emph{Distribution of the eigenvalues of random band matrices in the limit of their infinite order} (Russian) Teoret.\ Mat.\ Fiz. \textbf{90} (1992), no.\ 2, 163–178; translation in Theoret.\ and Math.\ Phys. \textbf{90} (1992), no.\ 2, 108–118.


\bibitem{psss}  R.\ Peled, J.\ Schenker, M.\ Shamis, S.\ Sodin: \emph{On the Wegner orbital model}, Int.\ Math.\ Res.\ Not. IMRN \textbf{2019}, no.\ 4, 1030–1058; arXiv:1608.02922.

\bibitem{schenker} J.\ Schenker: {\it Eigenvector Localization for Random Band Matrices with Power Law Band Width},  Commun.\ Math.\ Phys. {\bf 290}, 1065-1097 (2009).



\bibitem{ss2014} M.\ Shcherbina, T.\  Shcherbina: \emph{Characteristic polynomials for 1d random band matrices from the localization side},  Commun.\ Math.\ Phys. \textbf{351}, 1009–1044 (2017).

\bibitem{yy1} F.\  Yang, J.\  Yin: \emph{Random band matrices in the delocalized phase, III: averaging ﬂuctuations},  Probab.\ Theory Related Fields \textbf{179} (2021), no.\ 1-2, 451–540; arXiv:1807.02447

\end{thebibliography}
\end{document}